\documentclass[11pt]{article}
\usepackage{amssymb,amsfonts,amsmath,amsthm,bm}
\usepackage{algorithm2e,algorithmic,framed}

\usepackage{times}
\usepackage{graphicx} 
\usepackage{psfrag}
\usepackage{subfigure}
\usepackage{algorithmic}
\usepackage{hyperref}
\usepackage{amsthm,amsmath,amssymb}
\usepackage{bm} 
\usepackage{enumerate}
\usepackage{color}
\long\def\symbolfootnote[#1]#2{\begingroup%
\def\thefootnote{\fnsymbol{footnote}}\footnote[#1]{#2}\endgroup}

\newcommand{\FNorm }[1]{\mbox{}\|#1\|_\mathrm{F}  }

\newcommand{\TNorm }[1]{\mbox{}\|#1\|_2  }
\newcommand{\TNormS}[1]{\mbox{}\|#1\|_2^2}

\newtheorem{theorem}{\bf Theorem}[]
\newtheorem{lemma}[theorem]{Lemma}

\newcommand{\transp}{^{\textsc{T}}}

\newcommand{\mat}[1]{{\ensuremath{\bm{\mathrm{#1}}}}}

\newcommand{\pinv}[1]{ {#1}^{\dagger}}

\def\rank{\hbox{\rm rank}}

\def\b{{\mathbf b}}

\def\rb{{\mathbf r}}

\def\u{{\mathbf u}}
\def\v{{\mathbf v}}

\def\matA{\mat{A}}
\def\matB{\mat{B}}

\def\matE{\mat{E}}

\def\matI{\mat{I}}

\def\matQ{\mat{Q}}
\def\matR{\mat{R}}
\def\matS{\mat{S}}

\def\matU{\mat{U}}
\def\matV{\mat{V}}
\def\matW{\mat{W}}
\def\matX{\mat{X}}
\def\matY{\mat{Y}}
\def\matZ{\mat{Z}}
\def\matSig{\mat{\Sigma}}

\DeclareMathSymbol{\Prob}{\mathbin}{AMSb}{"50}
\newcommand\remove[1]{}
\newcommand\ignore[1]{}

\def\nnz{{ \rm nnz }}

\def\math#1{$#1$}

\def\mand#1{$$#1$$}
\def\frac#1#2{{#1\over #2}}

\def\mld#1{\begin{equation}
#1
\end{equation}}
\def\eqar#1{\begin{eqnarray}
#1
\end{eqnarray}}
\def\eqan#1{\begin{eqnarray*}
#1
\end{eqnarray*}}

\DeclareMathSymbol{\R}{\mathbin}{AMSb}{"52}

\def\qed{\hfill\rule{2mm}{2mm}}

\def\x{{\mathbf x}}

\def\z{{\mathbf z}}

\def\a{{\mathbf a}}
\def\b{{\mathbf b}}

\def\norm#1{{\|#1\|}}

\def\r#1{{(\ref{#1})}}

\newcommand{\red}[1]{{\color[named]{Red} #1}}
\newcommand{\blue}[1]{{\color[named]{Blue} #1}}

\def\dotfil{\leaders\hbox to 1.5mm{.}\hfill}
\newcounter{rmnum}
\def\RN#1{\setcounter{rmnum}{#1}\uppercase\expandafter{\romannumeral\value{rmnum}}}
\def\rn#1{\setcounter{rmnum}{#1}\expandafter{\romannumeral\value{rmnum}}}

\def\matLambda{\mat{\Lambda}}
\def\red#1{{\color[rgb]{1,0,0} #1}}
\def\blue#1{{\color[rgb]{0,0,1} #1}}
\def\nnz{{\text{nnz}}}

\begin{document}
\title{ {\bf Faster SVD-Truncated Regularized Least-Squares}}
\author{
Christos Boutsidis\thanks{
Mathematical Sciences Department, IBM T.J. Watson Research Center, Yorktown Heights NY 10598, USA.}
\and
Malik Magdon-Ismail\thanks{
Computer Science Department,
Rensselaer Polytechnic Institute, Troy NY 12180, USA,
magdon@cs.rpi.edu.}
}

\date{}
\maketitle

\begin{abstract}
We develop a fast algorithm for computing the
``SVD-truncated'' regularized  solution to the least-squares problem:
$ \min_{\x} \TNorm{\matA \x - \b}. $
Let $\matA_k$ of rank $k$ be the best rank $k$ matrix computed via the SVD of $\matA$.
Then, the SVD-truncated regularized solution is:
$ \x_k = \pinv{\matA}_k \b. $
If $\matA$ is $m \times n$, then,
it takes $O(m n \min\{m,n\})$ 
time
to compute $\x_k $ 
using the 
SVD of \math{\matA}.  We give an approximation algorithm
for \math{\x_k} which 
constructs a rank-\math{k} approximation
$\tilde{\matA}_{k}$ and computes 
$ \tilde{\x}_{k} = \pinv{\tilde\matA}_{k} \b$
in
roughly $O(\nnz(\matA) k \log n)$ 
time.
Our algorithm uses a randomized variant of the subspace iteration.
We show that, with high probability:
$ \TNorm{\matA  \tilde{\x}_{k} - \b} \approx \TNorm{\matA \x_k - \b}$
and
$\TNorm{\x_k - \tilde\x_k} \approx 0. $
\end{abstract}
\section{Introduction}
We consider the least-squares regression problem:
$$
\min_{\x \in \R^n}\TNorm{\matA \x -\b}, \hspace{.21in} 
\matA \in \R^{m \times n}, \hspace{.11in} 
\b\in\R^m, \quad m \ge n.
$$
We assume \math{\matA} has rank \math{\rho\le n}.
Via the SVD 
(see Section~\ref{sec:pre}), 
we can decompose \math{\matA} as:
$$
\matA = \sum_{i=1}^{\rho} \sigma_i \u_i  \v_i\transp,
$$
where
$\{ \u_i \in \R^m, \v_i \in \R^n, \sigma_i \in \R_+ \}$ are respectively
the left and  right singular vectors, and singular values 
of~\math{\matA}. The orthonormal
left and right singular matrices \math{\matU_\matA=[\u_1,\ldots,\u_{\rho}]} and 
\math{\matV_\matA=[\v_1,\ldots,\v_{\rho}]} have the singular vectors as
columns.
The optimal solution to 
this least-squares problem is:
$$ \x_{opt} = \sum_{i=1}^\rho \frac{\u_i\transp \b}{\sigma_i} \v_i. $$
When $\sigma_i \rightarrow 0,$ for some $i$, 
the solution is numerically unstable, and the problem becomes
ill-posed~\cite{Han90}.
Regularization helps with the numerical instablity as well as 
improving the generalization performance of machine learning algorithms
that use regression. 
Perhaps the simplest and most popular regularization technique
is 
Tikhonov regularization (or weight decay) 
\cite[Ex. 4.5]{LFDbook}
which results in the solution:
$$ \x_{\lambda} = \sum_{i=1}^{\rho} \frac{\sigma_i^2}{\sigma_i^2 + \lambda_i^2} \frac{\u_i\transp \b}{\sigma_i} \v_i, $$
where $\lambda_i > 0$ are the
regularization parameters (often chosen uniform, \math{\lambda_i=\lambda}). 
This regularized solution minimizes a penalized error
$
\TNormS{\matA \x -\b} + \TNormS{\matLambda\matV_{\matA} \transp\x},
$
where \math{\matLambda} is a $\rho \times \rho$ diagonal matrix with entries
\math{\matLambda_{ii}=\lambda_i}; each \math{\lambda_i} quantifies how much
one chooses to regularize the \math{i}th singular space of \math{\matA}.

SVD-truncated regularization, the focus of this paper,
 is a special case of Tikhonov regularization with 
\math{\lambda_i=0} for \math{i\le k} ($k< \rho$ is the truncation parameter)
and 
\math{\lambda_i\rightarrow\infty} otherwise~\cite{Han90}.
The SVD-truncated regularized solution \math{\x_k} is:
$$ \x_{k} = \sum_{i=1}^k  \frac{\u_i\transp \b}{\sigma_i} \v_i. $$
From the SVD,
$
\matA_k = \sum_{i=1}^k \sigma_i \u_i \v_i\transp
$ 
is
a best rank \math{k} approximation to 
\math{\matA}. So,
$\x_k$ solves a least-squares problem with \math{\matA_k}:
$$
\min_{\x \in \R^n}\TNorm{\matA_k \x -\b}, \hspace{.21in} 
\matA \in \R^{m \times n}, \hspace{.11in} 
\b\in\R^m, \quad m \ge n.
$$
Appropriately choosing \math{\lambda_i} or \math{k} are important
problems from the numerical linear algebra perspective as well as the 
machine learning perspective, and we refer to  
Section 5 in~\cite{Han90} for some discussion on this topic. For our
purposes, we take \math{k} as given.
That is, \math{\x_k} is the solution we want\footnote{When \math{\b} 
is concentrated in the top singular subspaces,
\math{\x_k}  approximates
\math{\x_{opt}}. Indeed, if
\math{\u_i\transp\b\ge\sigma_i^\alpha} for some \math{\alpha\ge 1},
then (\cite[Theorem 3.1]{Han90}): 
$
\TNorm{\x_{opt} - \x_k} / \TNorm{\x_{opt}}  \le 
                        \sqrt{n} \left( \sigma_{k+1} / \sigma_k \right)^{\alpha-1}.$ That is,
the SVD-truncated solution is
near optimal when 
the singular value gap \math{\gamma_k=\sigma_{k+1}/\sigma_k} is small.
},
and our goal is to
compute a good approximation to \math{\x_k} quickly in 
\math{o(mn^2)} time, since the SVD may be too expensive 
if $\matA$ is massive.

{\bf Our contributions.}
Via a recent randomized variant~\cite{RST09,HMT11} of the 
subspace iteration method, we develop a fast randomized algorithm
to compute an $\tilde{\x}_k$ in roughly $O(\nnz(\matA)k \log n)$ 
time where \math{\nnz(\matA)} is the number of non-zeros in \math{\matA}.
We describe this algorithm in Section~\ref{sec:main} and
give precise error estimates for its performance in Theorem~\ref{thm1}.
We show that there is not much room for improvement 
upon these estimates by providing a lower bound in Theorem~\ref{thm2}. 

\section{Related Work}
SVD-truncated regression
has been around for some time. 
See, for example,~\cite{Var73, Han90}
and references therein for some background and
applications of this regularization technique.

To develop faster SVD-truncated regression,
our approach is to first compute \math{\tilde\matA_k}, an approximation to 
\math{\matA_k}, obliviously to $\b$, and use \math{\tilde\matA_k} in the regression.
To construct \math{\tilde\matA_k}, we use 
an algorithm that was previously proposed to  quickly
construct a ``good''
low-rank approximation to a matrix in the spectral norm.
This algorithm is based on the
subspace iteration method~\cite[Sec 8.2.4]{GV12}
and was analyzed in~\cite{RST09},~\cite{HMT11}.

The approach of first approximating \math{\matA_k} is natural and has been
used before, 
for example \cite{xiang2013regularization} for uniform Tikhonov
regularization (\math{\lambda_i=\lambda}).
The main algorithm in~\cite{xiang2013regularization} is similar to ours
without the power iteration, and corresponds
to a random embedding into \math{k+q} dimensions 
before computing 
an approximate basis for the column space of \math{\matA}.
Theorem~1 in
\cite{xiang2013regularization} provides a high probability
bound:
\remove{
$$\frac{\norm{\x_\lambda-\hat\x_\lambda}_2^2}{\norm{\x_\lambda}_2^2}
=O\left((k+q)(q\log q+n-k)\frac{\sigma_{k+1}^2(\matA)}{\sigma_1^2(\matA)}\right).
$$
}
$$\norm{\x_\lambda-\hat\x_\lambda}_2^2
=O\left((k+q)(q\log q+n-k)\gamma_k^2
\right)\cdot{\norm{\x_\lambda}_2^2}
.
$$
Here, \math{\gamma_k=\sigma_{k+1}(\matA)/\sigma_k(\matA)\le1}.
This bound is similar in spirit to our Eqn.~\r{eqnthm2}
in Theorem~\ref{thm1}, except we work with 
SVD-truncated regularization, not uniform Tikhonov regularization,
and we give a stronger \math{O(\varepsilon)} bound.

The approach in~\cite{Vog94} also uses subspace iteration as we do,
with a different choice for the dimension reduction and an orthonormalization
step (see Section 3.1 in~\cite{Vog94}) - this choice is the ``classical subspace
iteration method'' from the numerical linear algebra literature~\cite[Sec 8.2.4]{GV12}. 
However, no theoretical bounds are reported in~\cite{Vog94}.
Iterative SVD-based methods such as the Lanczos iteration were also proposed
in Section 4 in~\cite{Vog94} and~\cite{OS81}. 
These approaches enjoy good empirical behavior, but again,
no theoretical bounds are known. 

In the above two results~\cite{Vog94, OS81}, when we say there are no theoretical 
bounds we mean there are no bounds for the regression setting, as those we provide
in Theorem~\ref{thm1}. However, subspace iteration and Lanczos iteration were extensively 
analyzed before and bounds similar to Lemma~\ref{lem1} are available.

An alternative approach to SVD-truncation
is feature selection or sparsity. 
In this setting, one selects columns from $\matA$ and solves the 
reduced regression with only these columns, resulting in a sparse
solution. See Section 12.2 in~\cite{GV12} for a discussion
of this approach. 
In recent work~\cite{BM13}, we developed a method based on column sampling that
runs 
in $O(m n \min\{m,n\} + n k^3 / \varepsilon^2)$ time and
returns a solution $\hat\x_r \in \R^n$ with 
$r = O(k)$ non-zero entries such that:
$$
\TNorm{\matA \hat\x_r-\b}
\le
\TNorm{ \matA \x_k - \b} + O(1)\norm{\b}_2\norm{\matA-\matA_k}_{\mathrm{F}} \sigma^{-1}_{k}(\matA).
$$
Eqn.~\ref{eqnthm1} in Theorem~\ref{thm1} in the present article, when we remove the
sparsity constraint, is considerably tighter.

A similar bound can be 
obtained
using the Rank-Revealing QR (RRQR) factorization~\cite{CH92a}:
a QR-like 
decomposition is used to select exactly $k$ columns of $\matA$
to obtain a 
sparse solution $\hat\x_k$.
Combining
Eqn.~(12) of~\cite{CH92a} with Strong RRQR~\cite{GE96} 
gives
$$ \TNorm{\x_k - \hat\x_k}  \le 3 \left( \sqrt{4 k(n-k)+1} \right) \sigma_k^{-1}(\matA) \cdot \TNorm{\b}.
$$
Eqn.~\ref{eqnthm2} in Theorem~\ref{thm1} in the present article, when we remove the
sparsity constraint, is considerably tighter.

\section{Preliminaries}\label{sec:pre}
\noindent{\em Basic Notation.}
We use \math{\matA,\matB,\ldots} to denote matrices;
\math{\a,\b,\ldots} to denote column vectors.
$\matI_{n}$ is the $n \times n$
identity matrix;  $\bm{0}_{m \times n}$ is the $m \times n$ matrix of zeros.
We use 
\math{\norm{\matA}_{\mathrm{F}}} for the Frobenius matrix norm 
and \math{\norm{\matA}_2} for the spectral or operator norm:
$ \FNorm{\matA}^2 = \sum_{i,j} \matA_{ij}^2$ and
$\TNorm{\matA} = \max_{\x:\TNorm{\x}=1}\TNorm{\matA \x}$.
By submultiplicativity,
$\TNorm{\matA \matB} \le \TNorm{\matA} \TNorm{\matB}$,
for any $\matA, \matB$. 

\noindent{\em
Singular Value Decomposition and the Pseudo-inverse.} \label{chap24}
The thin (compact) Singular Value Decomposition (SVD) of the matrix $\matA$ with $\rank(\matA) = \rho$ is:
\begin{eqnarray*}
\label{svdA} \matA
         = \underbrace{\left(\begin{array}{cc}
             \matU_{k} & \matU_{\rho-k}
          \end{array}
    \right)}_{\matU_{\matA} \in \R^{m \times \rho}}
    \underbrace{\left(\begin{array}{cc}
             \matSig_{k} & \bf{0}\\
             \bf{0} & \matSig_{\rho - k}
          \end{array}
    \right)}_{\matSig_\matA \in \R^{\rho \times \rho}}
    \underbrace{\left(\begin{array}{c}
             \matV_{k}\transp\\
             \matV_{\rho-k}\transp
          \end{array}
    \right)}_{\matV_\matA\transp \in \R^{\rho \times n}},
\end{eqnarray*}
where \math{\matSig_{\matA}}, a positive diagonal matrix, contains the
singular values in decreasing order: \math{(\matSig_{\matA})_{ii}=\sigma_i(\matA)} 
(we will drop the dependence on \math{\matA} and
use \math{\sigma_i} when the context is 
clear). 
The matrices $\matU_k \in \R^{m \times k}$ and $\matU_{\rho-k} \in \R^{m \times (\rho-k)}$ contain the left singular vectors of~$\matA$; similarly, 
$\matV_k \in \R^{n \times k}$ and $\matV_{\rho-k} \in \R^{n \times (\rho-k)}$ contain the right singular vectors of~$\matA$. 
$\matA_k=\matU_k \matSig_k \matV_k\transp = \matU_k \matU_k\transp\matA = \matA\matV_k \matV_k\transp \in \R^{m \times n}$ 
minimizes \math{\TNorm{\matA - \matX}} over all
matrices \math{\matX \in \R^{m \times n}} of rank at most $k$. 
Note that $\TNorm{\matA} = \sigma_1(\matA)$ and
$\TNorm{\matA-\matA_k} = \sigma_{k+1}(\matA)$. The pseudo-inverse of $\matA$ is
$\pinv{\matA} = \matV_\matA \matSig_\matA^{-1} \matU_\matA\transp 
\in \R^{n \times m}$.
The spectral gap of $\matA$ at \math{k < \rank(\matA)} is \math{\gamma_k=\sigma_{k+1}(\matA)/\sigma_k(\matA)\le1}.

\noindent{\em
Perturbation Theory.}
There exist bounds
on the perturbation of the pseudoinverse and singular values
 of a matrix
upon additive perturbation. Let $\matA, \matB, \matE$ be
 $m \times n$ matrices with $\matB = \matA + \matE$.  
\begin{lemma}[{\cite[Theorem 3.4]{stewart1977perturbation}}]\label{sec:pre1}
If $m \ge n$ and
$\rank(\matA) = \rank(\matB) < \min\{m,n\}:$  
$
\TNorm{\pinv{\matB} - \pinv{\matA}} \le 
2\TNorm{\pinv{\matA}}  \TNorm{\pinv{\matB}}\TNorm{\matE}.
$
\end{lemma}
\remove{
\begin{lemma}[{\cite[Remark 3.2]{xu2011optimal}}]\label{sec:pre1}
$$
\TNorm{\pinv{\matB} - \pinv{\matA}} \le \left(  \TNorm{\pinv{\matA}}  \TNorm{\pinv{\matB}} + \max\{  \TNormS{\pinv{\matA}}, \TNormS{\pinv{\matB}} \} \right) \TNorm{\matE}.
$$
Further, if $\rank(\matA) = \rank(\matB) < \min\{m,n\}$ and $m \ge n,$ 
then \cite[Theorem 3.4]{stewart1977perturbation}:
$$
\TNorm{\pinv{\matB} - \pinv{\matA}} \le 
2\TNorm{\pinv{\matA}}  \TNorm{\pinv{\matB}}\TNorm{\matE}
$$
\end{lemma}
}

\begin{lemma}[Weyl's inequality {\cite[Corollary 7.3.8]{HJ85}}]\label{lem:pert2}
$
| \sigma_i\left(\matB\right)- \sigma_i\left(\matA\right)  |  \le \TNorm{\matE},
\quad
\text{for } i=1,2,\ldots, \min( m,n).
$
\end{lemma}

\noindent{\em Random Matrix Theory.}\label{sec:random}
There exist results bounding the top and bottom singular
values of a random Gaussian matrix.
\begin{lemma}[Norm of a Gaussian Matrix~\cite{DavidsonSzarek01}]\label{lem:gauss1}
Let $\matX \in \R^{n \times m}$ be a matrix of i.i.d. standard Gaussian random variables, where $n \geq m.$ 
Then, for $t \geq 4$,
$
\Prob\{ \sigma_1(\matX) \ge t n^{\frac12} \} \ge e^{-n t^2/8}.
$
\end{lemma}
\begin{lemma}[Invertibility of a Gaussian Matrix~\cite{smoothedanalysis06}]\label{lem:gauss2}
Let $\matX \in \R^{n \times n}$ be a matrix with i.i.d. standard Gaussian random variables. Then,
for $\delta > 0,$
$
\Prob\{ \sigma_n(\matX) \le \delta n^{-\frac12} \}\leq 2.35 \delta.
$
\end{lemma}


\section{Main Result}~\label{sec:main}
We use an approximation
\math{\tilde\matA_k} (instead of \math{\matA_k}) 
and minimize \math{\norm{\tilde\matA_k\x-\b}_2} over \math{\x}. 
The algorithm is summarized below.
\begin{enumerate}
\item 
Compute \math{\matQ \in \R^{m \times k}}, an orthonormal basis 
for the columns of 
$
(\matA \matA\transp)^p \matA \matS \in \R^{m \times k},
$
where \math{p\ge0} and \math{\matS}
is an \math{n\times k} matrix of i.i.d. standard Gaussians.
\item 
Compute \math{\tilde\matA_k=\matQ\matQ\transp\matA} and 
\math{\tilde\x_k=\pinv{\tilde\matA}_k\b}.
\end{enumerate}
Careful implementation 
makes the algorithm efficient.
In Step~1, we compute the matrix
products in  \math{(\matA \matA\transp)^p \matA \matS} from right
to left to ensure a running time of \math{O(\nnz(\matA)kp)}; the result is 
an \math{m\times k} matrix, and a QR-factorization in 
 time \math{O(mk^2)}
gives
\math{\matQ}.
In Step 2, we need the SVD 
\math{\tilde\matA_k=\tilde\matU_k\tilde\matSig_k\tilde\matV_k\transp}.
Instead, we compute the SVD of 
\math{\matQ\transp\matA=\matU_{\matQ\transp\matA}\matSig_{\matQ\transp\matA}\matV\transp_{\matQ\transp\matA}}, in
\math{O(\nnz(\matA)k+nk^2)} time\footnote{It takes $O(\nnz(\matA)k)$ time for the matrix multiplication $\matQ\transp\matA,$
and $O(nk^2)$ time to compute the SVD of $\matQ\transp\matA$.}. 
Then, 
\math{\tilde\matA_k=\matQ\matU_{\matQ\transp\matA}\matSig_{\matQ\transp\matA}\matV\transp_{\matQ\transp\matA}}, from which we  read off the 
SVD of~\math{\tilde\matA_k} because \math{\matQ\matU_{\matQ\transp\matA}}
is orthonormal:
\math{\tilde\matU_k=\matQ\matU_{\matQ\transp\matA},\ 
\tilde\matSig_k=\matSig_{\matQ\transp\matA},\ 
\tilde\matV_k=\matV_{\matQ\transp\matA}.
}
Now, 
\math{\pinv{\tilde\matA}_k
=\tilde\matV_k\tilde\matSig_k^{-1}\tilde\matU_k\transp} and 
$
\tilde\x_k=\tilde\matV_k\tilde\matSig_k^{-1}\tilde\matU_k\transp\b,
$
which is computed in \math{O(mk+nk+k^2)} time. The dominant terms in the running
time 
are \math{O(\nnz(\matA)kp+(m+n)k^2)}.

We control the accuracy of the
algorithm by choosing \math{p} appropriately. 
A larger \math{p} gives a better error.
Let \math{0 < \varepsilon < 1} be an error parameter and 
recall that the spectral gap of
\math{\matA} at \math{k} is 
$
\gamma_k= \sigma_{k+1}(\matA) / \sigma_{k}(\matA) \le 1.
$
The next theorem quantifies how the error depends on \math{p}. Roughly
speaking, setting
$
p=O\left({ \ln(\varepsilon/n)}/{\ln(\gamma_k)} \right)
$
suffices
to give additive
error \math{\varepsilon\norm{\b}_2}. 

\begin{theorem}\label{thm1}
Fix \math{\matA\in\R^{m\times n},
\ \b\in\R^m},  $k < \rank(\matA)$, and \math{\varepsilon,\delta\in(0,1)}.
Choose \math{p} in our algorithm to satisfy
$$
p\ge \frac{ \ln( \varepsilon \cdot \delta \cdot  \frac{\sigma_k^2} {\sigma_1^2} \cdot \frac{1}{12n} )} {\ln\left(\gamma_k^2\right)}.
$$
Let 
\math{\tilde\x_k=\pinv{\tilde\matA}_k\b} and $\x_k = \pinv{\matA}_k \b$ be the exact SVD-truncated solution.
Then, with probability at least $1- e^{-2n} - 2.35\delta:$
\begin{equation}\label{eqnthm1}
\TNorm{\matA \tilde\x_{k} - \b} \le \TNorm{\matA \x_{k} - \b} +\varepsilon \cdot \TNorm{\b},
\end{equation}
and
\begin{equation}\label{eqnthm2}
\frac{\TNorm{\x_k - \tilde\x_k}}{\TNorm{ \x_k }} \le 
\frac{4}{3}
\cdot 
\varepsilon.
\end{equation}
\end{theorem}	
The error in~\r{eqnthm1} is additive, and in Section~\ref{sec:lower} we
show that this 
is unavoidable when, as we do, one solves the regression via
an approximation \math{\tilde\matA_k} which is constructed obliviously to~$\b$.

\subsection{Proof of Theorem~\ref{thm1}}

Recall that \math{\matA_k=\matU_k\matU_k\transp\matA} and
\math{\x_k=\pinv{\matA}_k\b}. By our construction of $\tilde\matA_k$
and $\tilde\matU_k$,
\math{\tilde\matA_k=\tilde\matU_k\tilde\matU_k\transp\matA}, with 
\math{\tilde\x_k=\pinv{\tilde\matA}_k\b}. We first
quantify the additive error. By the 
triangle inequalty,
\eqan{
\TNorm{\matA \tilde\x_{k} - \b}
&\le&
\TNorm{ \matA\x_k  - \b}+\Delta,
}
where 
$$
\Delta=\TNorm{\matA \tilde\x_{k} - \matA\x_k}=
\norm{\matA(\pinv{\tilde\matA}_{k} - \pinv{\matA}_k) \b}_2. 
$$
We need to
upper bound \math{\Delta}. By submultiplicativity,
\eqar{
\TNorm{\tilde\matA_{k}-\matA_{k}}
&=&
\TNorm{(\tilde\matU_k\tilde\matU_k\transp-\matU_k\matU_k\transp)\matA}
\nonumber\\
&\le&
\TNorm{\matA}\TNorm{\tilde\matU_k\tilde\matU_k\transp-\matU_k\matU_k\transp}
\label{eq:perturb1}. 
}
\begin{lemma}\label{lem:basic1}
\math{\Delta\le
\frac{2\sigma_1^2(\matA)}{\sigma_k(\tilde\matA)\sigma_k(\matA)}
\TNorm{\tilde\matU_k\tilde\matU_k\transp-\matU_k\matU_k\transp} 
\TNorm{\b}
}.
\end{lemma}
\begin{proof} We manipulated $\Delta$ as follows: 
\eqan{
\Delta
& = & 
\TNorm{\matA(\pinv{\tilde\matA}_{k}-\pinv{\matA}_{k}) \b}\\
& \le & 
\TNorm{\matA}\TNorm{\pinv{\tilde\matA}_{k}-\pinv{\matA}_{k}}\TNorm{\b}\\
& \le & 
2\TNorm{\matA}\TNorm{\pinv{\tilde\matA}_{k}}
\TNorm{\pinv{\matA}_{k}}\TNorm{\tilde\matA_{k}-\matA_{k}}\TNorm{\b}\\
& \le & 
\frac{2\sigma_1^2(\matA)}{\sigma_k(\tilde\matA)\sigma_k(\matA)}
\TNorm{\tilde\matU_k\tilde\matU_k\transp-\matU_k\matU_k\transp} 
\TNorm{\b}\\
}
The first inequality uses submultiplicativity; 
the second uses Lemma~\ref{sec:pre1}; and, the last
uses Eqn.~\r{eq:perturb1}. 
\end{proof}
Lemma~\ref{lem:basic1} holds no matter what $\tilde{\matU}_k$ is.
The difference in the projection operators
\math{\TNorm{\tilde\matU_k\tilde\matU_k\transp - \matU_k\matU_k\transp}}
plays an important role in our bounds.
Our algorithm constructs $\tilde{\matU}_k$ for which this error term can
be bounded.
A similar application of the power iteration was analyzed for spectral
clustering in~\cite{GKB13}.
For the specific 
\math{\tilde\matU_k} returned by 
our algorithm, the difference in the projection operators 
can be bounded with high probability.
\begin{lemma}~\cite{GKB13}\label{lem1}
Fix \math{\varepsilon,\delta\in(0,1)}.
If \math{p\ge\ln(\varepsilon\delta/4n)/\ln(\gamma_k^2)}, then with
probability at least \math{1-e^{-2n}-2.35\delta},
$$
\TNorm{\tilde\matU_k\tilde\matU_k\transp - \matU_k\matU_k\transp} 
\le 
\varepsilon.
$$
\end{lemma}
\begin{proof}
See the Appendix.
\end{proof}
Lemma~\ref{lem1} bounds the difference in the projection operators.
Notice that in Lemma~\ref{lem:basic1} we also need the $k$th
singular value of \math{\tilde\matA_k}. This can be bounded by
Weyl's theorem~(Lemma~\ref{lem:pert2}):
\eqan{
|\sigma_k(\tilde\matA_k)-\sigma_k(\matA_k)|
&\le&
\norm{\tilde\matA_k-\matA_k}_2
\\
&\le&
\sigma_1(\matA)\TNorm{\tilde\matU_k\tilde\matU_k\transp - 
\matU_k\matU_k\transp},
}
from which we have that
\mld{
\sigma_k(\tilde\matA_k)\ge
\sigma_k
-\sigma_1\TNorm{\tilde\matU_k\tilde\matU_k\transp - 
\matU_k\matU_k\transp}.
\label{eq:singular1}
}
We are now ready to prove Eqn.~\r{eqnthm1} in 
Theorem~\ref{thm1}. Set
\mand{
p\ge
\frac{\ln\left(\varepsilon\delta\sigma_k^2/12n\sigma_1^2\right)}{
\ln(\gamma_k^2)}
=
\frac{\ln\left((\frac{\varepsilon\sigma_k^2}{3\sigma_1^2})\delta/4n\right)}{
\ln(\gamma_k^2)}.
}
It now follows from Lemma~\ref{lem1} that 
\mld{
\TNorm{\tilde\matU_k\tilde\matU_k\transp - \matU_k\matU_k\transp}
\le
\frac{\varepsilon\sigma_k^2}{3\sigma_1^2}.\label{eq:project1}
}
From \r{eq:singular1}, it follows that 
\eqar{
\sigma_k(\tilde\matA_k)
\ge
\sigma_k
-\frac{\varepsilon\sigma_k^2}{3\sigma_1}
=
\sigma_k\left(1-\frac{\varepsilon\sigma_k}{3\sigma_1}\right)
\ge
\frac23\sigma_k.\label{eq:singular2}
}
Using Lemma~\ref{lem:basic1} with \r{eq:project1} and \r{eq:singular2}
we obtain a bound for \math{\Delta}:
\mand{
\Delta\le\frac{2\sigma_1^2}{\frac23\sigma_k\cdot\sigma_k}
\cdot
\left(\frac{\varepsilon\sigma_k^2}{3\sigma_1^2}
\right)\cdot\norm{\b}_2
=\varepsilon\norm{\b}_2.
}
We move to the proof of Eqn.~\r{eqnthm2} in Theorem~\ref{thm1}.
We need a perturbation theory result from~\cite{Han87} which
we state in our notation for a perturbation
of the matrix
\math{\matA_k} to the matrix 
\math{\tilde\matA_k}, without any perturbation on the response~\math{\b}.
Let   
$$
\matE=\tilde\matA_k-\matA_k.
$$
\begin{lemma}[{\cite[Eqn. (27)]{Han87}}]\label{lem:Han87}
If  
\math{\norm{\matE}_2<\sigma_k}, then,
\mand{
\frac{\norm{\x_k-\tilde\x_k}_2}{\norm{\x_k}_2}
\le
\frac{\sigma_1\norm{\matE}_2}{\sigma_k-\norm{\matE}_2}\left(
\frac{1}{\sigma_1}+\frac{\norm{\matA_k\x_k-\b}_2}{\sigma_k\norm{\b}_2}
\right)+\frac{\norm{\matE}_2}{\sigma_k}.
}
\end{lemma}
We can simplify the bound in Lemma~\ref{lem:Han87}
because \math{\norm{\matA_k\x_k-\b}_2\le\norm{\b}_2}:
\mld{
\frac{\norm{\x_k-\tilde\x_k}_2}{\norm{\x_k}_2}
\le
\frac{\sigma_1\norm{\matE}_2}{\sigma_k-\norm{\matE}_2}\left(
\frac{1}{\sigma_1}+\frac{1}{\sigma_k}
\right)+\frac{\norm{\matE}_2}{\sigma_k}.\label{eq:Han87}
}
Using the bound \r{eq:project1} in \r{eq:perturb1} and recalling that
\math{\gamma_k=\frac{\sigma_{k+1}}{\sigma_k}},
\mand{
\norm{\matE}=\norm{\tilde\matA_k-\matA_k}_2\le\frac{\varepsilon\gamma_k}{3}\cdot\sigma_k.
}
In particular, since \math{\varepsilon<1} and \math{\gamma_k\le1},
\math{\norm{\tilde\matA_k-\matA_k}_2<\sigma_k} and we can apply
Lemma~\ref{lem:Han87}, or the bound in Eqn. \r{eq:Han87}:
\eqan{
\frac{\norm{\x_k-\tilde\x_k}_2}{\norm{\x_k}_2}
&\le&
\frac{\sigma_1(\frac{\varepsilon\gamma_k}{3}\sigma_k)}{\sigma_k-(\frac{\varepsilon\gamma_k}{3}\sigma_k)}\left(
\frac{1}{\sigma_1}+\frac{1}{\sigma_k}
\right)+\frac{(\frac{\varepsilon\gamma_k}{3}\sigma_k)}{\sigma_k}\\
&=&
\frac{\varepsilon\gamma_k}{3}
\cdot
\left(
\frac{1}{1-\frac{\varepsilon\gamma_k}{3}}\left(
1+\frac{1}{\gamma_k}
\right)+1
\right)\\
&\le&
\frac{\varepsilon\gamma_k}{3}
\cdot
\left(
\frac{3}{2}\left(
1+\frac{1}{\gamma_k}
\right)+1
\right)\\
&=&
\frac{\varepsilon}{3}
\cdot
\left(\frac{5\gamma_k}{2}+\frac32\right)
\le \frac{4}{3}\varepsilon.
}
(The second inequality is because 
\math{1-\frac{\varepsilon\gamma_k}{3}\ge\frac23}; and,
the final inequality is because \math{\gamma_k\le 1}.)
\qed

\section{Lower Bound: Additive Error is Unavoidable}\label{sec:lower}

We now show that the additive error of Eqn.~\ref{eqnthm1} in Theorem~\ref{thm1} is tight.
Towards this end, let us consider the class of (fast) algorithms which operate as follows:
\begin{enumerate}\itemsep0pt
\item Quickly construct matrix \math{\tilde\matA_k} of rank \math{k} obliviously to $\b$.
\item Use \math{\tilde \matA_k} to construct the approximate solution
\math{\tilde \x_k=\pinv{\tilde\matA}_k \b}.
\end{enumerate}
Let \math{\norm{(\matI-\matA_k\pinv{\tilde\matA}_k)\matA_k}_2
=\varepsilon\norm{\matA_k}_2}.
The cross-`projection' operator \math{\matA_k\pinv{\tilde\matA}_k}
quantifies how well \math{\pinv{\tilde\matA}_k} approximates
\math{\pinv{\matA}_k} (if \math{\pinv{\tilde\matA}_k=\pinv{\matA}_k}, then
it is a projection operator and \math{\varepsilon=0}).
Note that
\math{\pinv{\tilde\matA}_k-\pinv{\matA}_k} and
\math{\matA_k-\tilde\matA_k=(\matU_k\matU_k\transp-\tilde\matU_k\tilde\matU_k\transp)
\matA} are related (see the discussion in Section~\ref{sec:pre1}),
and for our algorithm
\math{\norm{\matU_k\matU_k\transp-\tilde\matU_k\tilde\matU_k\transp}_2} is bounded by $\varepsilon$ (see Lemma~\ref{lem1}).

The next theorem
states that
the additive error in Eqn.~\ref{eqnthm1} in Theorem~\ref{thm1}
is about the best you
can expect of algorithms that construct \math{\tilde\x_k} via an
approximation
\math{\tilde \matA_k}, provided that \math{\tilde\matA_k} is constructed
obliviously to \math{\b}. The notion of approximation we consider is
via the equation  \math{\norm{(\matI-\matA_k\pinv{\tilde\matA}_k)\matA_k}_2
=\varepsilon\norm{\matA_k}_2}.
\begin{theorem}\label{thm2}
Fix \math{\matA,\tilde\matA_k\in\R^{m\times n}}. Let
\math{\matA_k} be the best
rank-\math{k} approximation to \math{\matA} used in the
top-\math{k} SVD-truncated regression and suppose
\math{\tilde\matA_k}
satisfies
(for $\varepsilon > 0$):
\mand{\norm{(\matI-\matA_k\pinv{\tilde\matA}_k)\matA_k}_2=\varepsilon\norm{\matA_k}_2.}
Then, for some $\b\in\R^{m}$, with $\x_k = \pinv{\matA}_k \b$ and
$\tilde\x_k = \pinv{\tilde\matA}_k \b$,
\eqan{
\norm{\matA\x_k-\b}_2&=&0\\
\norm{\matA\tilde\x_k-\b}_2&\ge&\epsilon\norm{\b}_2.
}
In particular, no multiplicative error 
bound is possible and the 
additive error is at least \math{\varepsilon \norm{\b}_2}.
\end{theorem}
\begin{proof}
We set \math{\b=\matA_k\z} for \math{\z} to be selected later. Then
\mand{\matA\x_k-\b=(\matA\pinv{\matA}_k\matA_k-\matA_k)\z=\bm{0}.}
(The last equality is because
\math{\matA\pinv{\matA}_k\matA_k=\matA_k\pinv{\matA}_k\matA_k=\matA_k}.)
We now manipulate \math{\norm{\matA\tilde\x_k-\b}_2}.
\eqan{
\norm{\matA\tilde\x_k-\b}_2
&=&
\norm{\matA\pinv{\tilde\matA}_k \matA_k\z-\matA_k\z}_2\\
&=&
\norm{\matA_k\pinv{\tilde\matA}_k \matA_k\z-\matA_k\z+
\matA_{\rho-k}\pinv{\tilde\matA}_k  \matA_k\z}_2\\
&\ge&
\norm{\matA_k\pinv{\tilde\matA}_k  \matA_k\z-\matA_k\z}_2\\
&=&
\norm{(\matI-\matA_k\pinv{\tilde\matA}_k )\matA_k\z}_2\\
}
(The inequality follows because \math{\norm{\matX+\matY}_2\ge\norm{\matX}_2}
when \math{\matX\transp\matY=\bm0}.)
We now choose \math{\z} to be the top right singular vector of the matrix
\math{(\matI-\matA_k\pinv{\tilde\matA}_k )\matA_k}. Then,
\eqan{
\norm{(\matI-\matA_k\pinv{\tilde\matA}_k )\matA_k\z}_2
&=&
\norm{(\matI-\matA_k\pinv{\tilde\matA}_k )\matA_k}_2
\cdot\norm{\z}_2\\
&\ge&
\norm{(\matI-\matA_k\pinv{\tilde\matA}_k )\matA_k}_2
\cdot\frac{\norm{\b}_2}{\norm{\matA_k}_2}\\
&=&\epsilon\norm{\b}_2.
}
The inequality uses
$
\norm{\b}_2=\norm{\matA_k\z}_2\le\norm{\matA_k}_2\norm{\z}_2,
$
and
the last equality is from the theorem statement.
\end{proof}

\section{Numerical Illustration}\label{sec:expt}
\renewcommand{\rb}[1]{\rotatebox{90}{#1}}
\providecommand{\st}[1]{\scalebox{1.35}{#1}}
\providecommand{\stn}[1]{\scalebox{1.1}{#1}}

\begin{figure*}[t]
\begin{tabular}{c@{\hspace*{0.1\textwidth}}c}
\resizebox{0.42\textwidth}{!}{
%
%
\begin{psfrags}%
\psfragscanon%
%
\psfrag{s05}[l][l]{\color[rgb]{0,0,0}\setlength{\tabcolsep}{0pt}\begin{tabular}{l}\st{Dimension $n$}\end{tabular}}%
\psfrag{s06}[l][l]{\color[rgb]{0,0,0}\setlength{\tabcolsep}{0pt}\begin{tabular}{l}\rb{\st{Error \%}}\end{tabular}}%
\psfrag{s07}[l][l]{\color[rgb]{0,0,0}\setlength{\tabcolsep}{0pt}\begin{tabular}{l}\st{\blue{$\displaystyle\frac{\norm{\matA\tilde\x_k-\b}_2-\norm{\matA\x_k-\b}_2}{\norm{\b}_2}$}}\end{tabular}}%
\psfrag{s08}[l][l]{\color[rgb]{0,0,0}\setlength{\tabcolsep}{0pt}\begin{tabular}{l}\st{\red{$\displaystyle\frac{\norm{\tilde\x_k-\x_k}_2}{\norm{\x_k}_2}$}}\end{tabular}}%
%
\psfrag{x01}[t][t]{\stn{100}}%
\psfrag{x02}[t][t]{\stn{300}}%
\psfrag{x03}[t][t]{\stn{500}}%
\psfrag{x04}[t][t]{\stn{700}}%
\psfrag{x05}[t][t]{\stn{900}}%
\psfrag{x06}[t][t]{\stn{1100}}%
\psfrag{x07}[t][t]{\stn{1300}}%
\psfrag{x08}[t][t]{\stn{1500}}%
%
\psfrag{v01}[r][r]{\stn{0}}%
\psfrag{v02}[r][r]{\stn{2}}%
\psfrag{v03}[r][r]{\stn{4}}%
\psfrag{v04}[r][r]{\stn{6}}%
\psfrag{v05}[r][r]{\stn{8}}%
%
\resizebox{12cm}{!}{\includegraphics{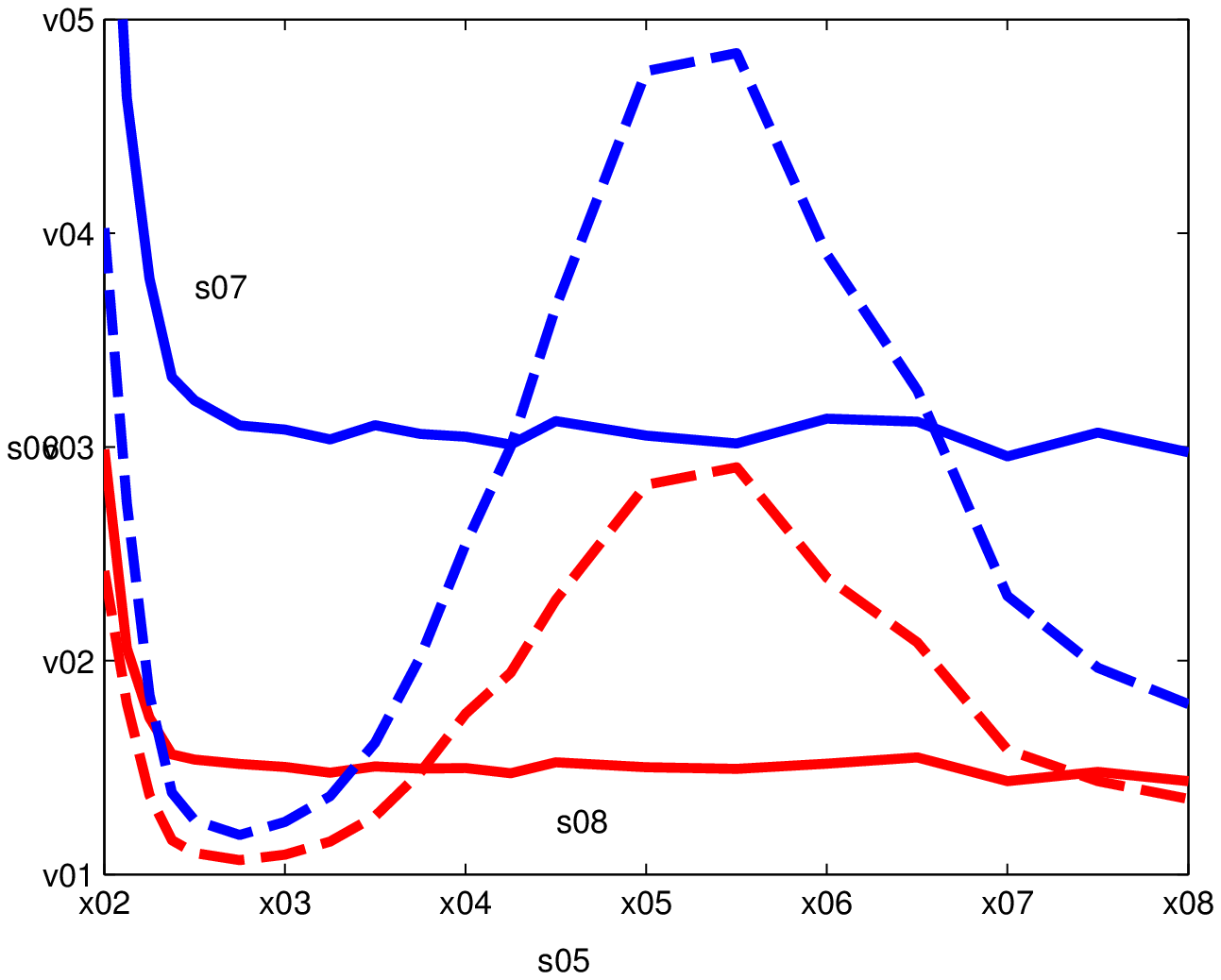}}%
\end{psfrags}%
%
}
&
\resizebox{0.42\textwidth}{!}{
%
%
\begin{psfrags}%
\psfragscanon%
%
\psfrag{s05}[l][l]{\color[rgb]{0,0,0}\setlength{\tabcolsep}{0pt}\begin{tabular}{l}\st{Dimension $n$}\end{tabular}}%
\psfrag{s06}[l][l]{\color[rgb]{0,0,0}\setlength{\tabcolsep}{0pt}\begin{tabular}{l}\rb{\st{Time($\x_k$) / Time($\tilde\x_k$)}}\end{tabular}}%
%
\psfrag{x01}[t][t]{\stn{100}}%
\psfrag{x02}[t][t]{\stn{300}}%
\psfrag{x03}[t][t]{\stn{500}}%
\psfrag{x04}[t][t]{\stn{700}}%
\psfrag{x05}[t][t]{\stn{900}}%
\psfrag{x06}[t][t]{\stn{1100}}%
\psfrag{x07}[t][t]{\stn{1300}}%
\psfrag{x08}[t][t]{\stn{1500}}%
%
\psfrag{v01}[r][r]{\stn{0}}%
\psfrag{v02}[r][r]{\stn{0.25}}%
\psfrag{v03}[r][r]{\stn{0.5}}%
\psfrag{v04}[r][r]{\stn{0.75}}%
\psfrag{v05}[r][r]{\stn{1.0}}%
\psfrag{v06}[r][r]{\stn{1.25}}%
\psfrag{v07}[r][r]{\stn{1.5}}%
%
\resizebox{12cm}{!}{\includegraphics{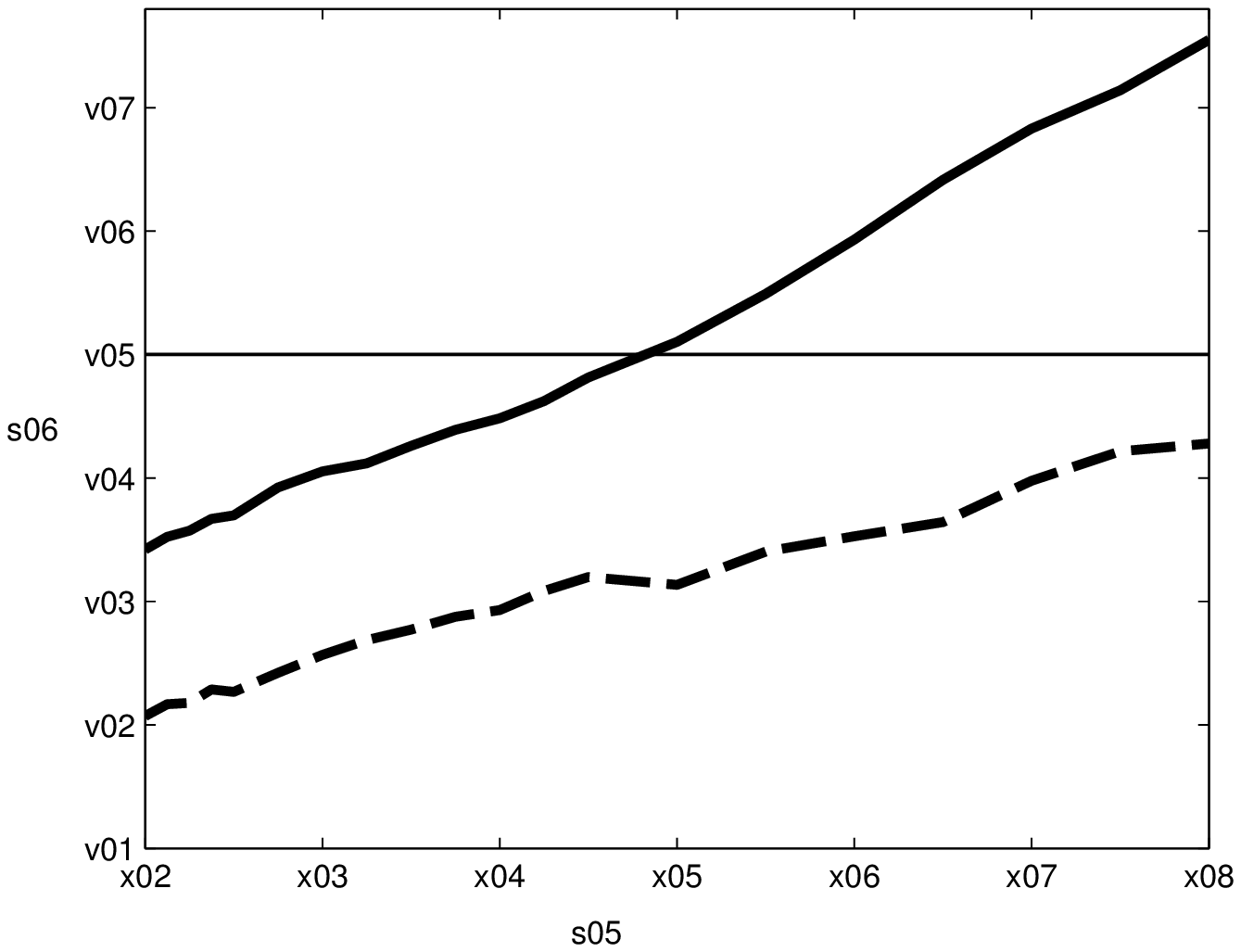}}%
\end{psfrags}%
%
}
\\
(a) Accuracy&(b) Running Time
\end{tabular}
\caption{(a) Solid curves: accuracy ratio for
our fast SVD-truncated regression algorithm,
\math{k=20}, \math{p=20\ln n}. Dashed curves (for comparison):
MATLAB's {\sf svds} solver which computes an approximation
to the top-\math{k} singular space; we set the error tolerance
to produce comparable error to our algorithm
({\sf options.tol=0.001}). The performance of our algorithm results in 
nearly constant error (approx 4\% in the objective and
1\% in the solution vector); the accuracy of {\sf svds} is approximately
the same but very unpredictible.
(b) Ratio of the time to compute the exact solution over time to compute
the approximation. For \math{n=1500} our algorithm is about 50\% more 
efficient. MATLAB's {\sf svds} is also more
efficient asymptotically than the exact solution, but is not as efficient as
our algorithm.
\label{fig:illustration}}
\end{figure*}

We perform a  numerical experiment on a synthetic
regression problem 
to illustrate the theory and the algorithm.
We construct a synthetic problem as follows.
We generate an \math{n\times n} matrix \math{\matA} of i.i.d. Gaussians,
and set the spectral gap \math{\gamma_k= \sigma_{k+1} / \sigma_k =0.99}. 
To do this,
use the SVD, \math{\matA=\matU\matSig\matV\transp} and 
rescale \math{\sigma_{k+1},\ldots,\sigma_n} up or down by 
a constant factor so that \math{\gamma_k=0.99}. 
Now reconstruct \math{\matA} using \math{\matU}, \math{\matV} and the 
rescaled~\math{\matSig}.
We construct the response
\math{\b=\frac{\matA_k\mathbf{r}_1\ }{\norm{\matA_k\mathbf{r}_1}_2}
+0.2\times\frac{\mathbf{r}_2\ }{\norm{\mathbf{r}_2}_2}},
where \math{\mathbf{r}_1,\mathbf{r}_2} are random standard Gaussian vectors.
So, the response \math{\b}
 has roughly \math{80\%} within the top-\math{k} singular
space.
We set 
\math{k=20} and run our algorithm with
\math{p=10\ln n}.
We vary \math{n\in[100,1000]} and for each value of \math{n} take the
average over several experiments to increase statistical significance.

For comparison, we also
use the truncated SVD algorithm {\sf svds} distributed
with MATLAB 8.1, where one can specify 
 an error tolerance {\sf tol}; {\sf svds} returns 
\math{\tilde\matU_k,\tilde\matSig_k,\tilde\matV_k} for which 
\math{\norm{\matA\tilde\matV_k-\tilde\matU_k\tilde\matSig_k}_2
\le\text{\sf tol}\cdot\norm{\matA}_2}.
The accuracy and running time results are
shown in Figure~\ref{fig:illustration} which illustrates the linear
speedup of our algorithm. For reference, at \math{n=1000}, the exact solution
takes about 2.5s on a single CPU laptop. Our algorithm performs according to
the theory (with \math{p=O(\ln n)} we achieve approximately 
fixed relative
error).

\bibliographystyle{abbrv}
\bibliography{RPI_BIB}

\begin{thebibliography}{10}

\bibitem{LFDbook}
Y.~S. Abu-Mostafa, M.~Magdon-Ismail, and H.-T. Lin.
\newblock {\em Learning from data}.
\newblock AMLBook, 2012.

\bibitem{BM13}
C.~Boutsidis and M.~Magdon-Ismail.
\newblock A note on sparse least-squares regression.
\newblock {\em Information Processing Letters}, 2013.

\bibitem{CH92a}
T.~F. Chan and P.~C. Hansen.
\newblock Some applications of the rank revealing {QR}~factorization.
\newblock {\em SIAM Journal on Scientific and Statistical Computing},
  13:727--741, 1992.

\bibitem{DavidsonSzarek01}
K.~R. Davidson and S.~J. Szarek.
\newblock {Local Operator Theory, Random Matrices and Banach Spaces}.
\newblock In {\em Handbook of the Geometry of Banach Spaces}, volume~1.
  Elsevier Science, 2001.

\bibitem{GKB13}
A.~Gittens, P.~Kambadur, and C.~Boutsidis.
\newblock Approximate spectral clustering via randomized sketching.
\newblock {\em preprint arXiv:1311.2854}, 2013.

\bibitem{GV12}
G.~Golub and C.~Van~Loan.
\newblock {\em Matrix computations}.
\newblock JHU Press, 2012.

\bibitem{GE96}
M.~Gu and S.~Eisenstat.
\newblock Efficient algorithms for computing a strong rank-revealing {QR}
  factorization.
\newblock {\em SIAM Journal on Scientific Computing}, 17:848--869, 1996.

\bibitem{HMT11}
N.~Halko, P.-G. Martinsson, and J.~A. Tropp.
\newblock Finding structure with randomness: Probabilistic algorithms for
  constructing approximate matrix decompositions.
\newblock {\em SIAM review}, 53(2):217--288, 2011.

\bibitem{Han87}
P.~Hansen.
\newblock The truncated {SVD} as a method for regularization.
\newblock {\em BIT Numerical Mathematics}, 27(4):534--553, 1987.

\bibitem{Han90}
P.~C. Hansen.
\newblock Truncated singular value decomposition solutions to discrete
  ill-posed problems with ill-determined numerical rank.
\newblock {\em SIAM Journal on Scientific and Statistical Computing},
  11(3):503--518, 1990.

\bibitem{HJ85}
R.~Horn and C.~Johnson.
\newblock {\em Matrix Analysis}.
\newblock CUP, NY, 1985.

\bibitem{OS81}
D.~P. O'Leary and J.~A. Simmons.
\newblock A bidiagonalization-regularization procedure for large scale
  discretizations of ill-posed problems.
\newblock {\em SIAM Journal on Scientific and Statistical Computing},
  2(4):474--489, 1981.

\bibitem{RST09}
V.~Rokhlin, A.~Szlam, and M.~Tygert.
\newblock A randomized algorithm for principal component analysis.
\newblock {\em SIAM Journal on Matrix Analysis and Applications},
  31(3):1100--1124, 2009.

\bibitem{smoothedanalysis06}
A.~Sankar, D.~A. Spielman, and S.-H. Teng.
\newblock Smoothed analysis of the condition numbers and growth factors of
  matrices.
\newblock {\em SIAM Journal on Matrix Analysis and Applications},
  28(2):446--476, 2006.

\bibitem{stewart1977perturbation}
G.~Stewart.
\newblock On the perturbation of pseudo-inverses, projections and linear least
  squares problems.
\newblock {\em SIAM review}, 19(4):634--662, 1977.

\bibitem{Var73}
J.~M. Varah.
\newblock On the numerical solution of ill-conditioned linear systems with
  applications to ill-posed problems.
\newblock {\em SIAM Journal on Numerical Analysis}, 10(2):257--267, 1973.

\bibitem{Vog94}
C.~Vogel and J.~Wade.
\newblock Iterative {SVD}-based methods for ill-posed problems.
\newblock {\em SIAM Journal on Scientific Computing}, 15(3):736--754, 1994.

\bibitem{xiang2013regularization}
H.~Xiang and J.~Zou.
\newblock Regularization with randomized {SVD} for large-scale discrete inverse
  problems.
\newblock {\em Inverse Problems}, 29(8):085008, 2013.

\end{thebibliography}


\appendix

\section{Proof of Lemma~\ref{lem1}}\label{sec:appproof} 
The result appeared in prior work~\cite[Corollary 11]{GKB13}.
Nevertheless, for completeness, we give a short, different
proof
based on~\cite[Theorem 2.6.1]{GV12}, which states that for any two $m \times k$ orthonormal
matrices \math{\matW,\matZ} with $m \ge k$: 
\mand{
\norm{\matW\matW\transp-\matZ\matZ\transp}_2=
\norm{\matZ\transp\matW^{\perp}}_2
=
\norm{\matW\transp\matZ^{\perp}}_2.
}
$\matZ^{\perp}\in \R^{m \times (m-k)}$ 
is such that $[\matZ, \matZ^{\perp}] \in \R^{m \times m}$ is a full orthonormal basis. We set \math{\matU_k^\perp=[\matU_{\rho-k},\matU_{m-\rho-k}]}. 

Given some (any) \math{\matS\in\R^{n\times k}},
recall that in our algorithm, \math{\matQ} is obtained by a 
QR-factorization of \math{(\matA\matA\transp)^p\matA \matS}:
\mand{
(\matA\matA\transp)^p\matA \matS=\matQ\matR.}
where 
\math{\matQ\in\R^{m\times k}} and 
\math{\matR\in\R^{k\times k}}. We need some basic facts:
\eqar{
\matQ\matR
&=&
\matU_k\matSig_k^{2p+1}\matV_k\transp\matS
+
\matU_{\rho-k}\matSig_{\rho-k}^{2p+1}\matV_{\rho-k}\transp\matS;
\ \ \ \ \ \
\label{app:basic1}\\
\sigma_k\left(\matQ\matR\right)
&\ge&
\sigma_k\left(\matU_k\matSig_k^{2p+1}\matV_k\transp\matS\right)
\ge
\sigma_k^{2p+1}\sigma_k(\matV_k\transp\matS)\label{app:basic2};\\
\sigma_i(\matQ\matR)&=&\sigma_i(\matR)\label{app:basic3};\\
\norm{\matX\matR}_2&\ge&\norm{\matX}_2\sigma_k(\matR),
\quad\text{for any }\matX\in\R^{\ell\times k}.
\label{app:basic4}
}
\r{app:basic1} follows from a direct computation using 
the SVD of \math{\matA};
\r{app:basic2} follows from 
\r{app:basic1} because \math{\matU_k} and \math{\matU_{\rho-k}}
span orthogonal spaces, and the fact that the minimum singular
value of a product is at least the product of the minimum singular
values;
\r{app:basic3} follows because \math{\matQ\transp\matQ=\matI_k};
\r{app:basic4} is well known: it is clear if 
\math{\sigma_k(\matR)=0} and if \math{\sigma_k(\matR)>0} then it follows from:
\mand{
\norm{\matX_2}
=
\max_{\x\not=\bm0}\frac{\norm{\matX\matR\x}_2}{\norm{\matR\x}_2}
\le
\max_{\x\not=\bm0}\frac{\norm{\matX\matR\x}_2}{\sigma_k(\matR)\norm{\x}_2}
=
\frac{\norm{\matX\matR}_2}{\sigma_k(\matR)}.
}
Observe that
\math{
\tilde\matU_k\tilde\matU_k\transp
=
\matQ\matU_{\matQ\transp\matA}\matU_{\matQ\transp\matA}\transp\matQ\transp
=
\matQ\matQ\transp,
}
because 
\math{\matU_{\matQ\transp\matA}\matU_{\matQ\transp\matA}=\matI_{k}}.
Therefore, using \cite[Theorem 2.6.1]{GV12},
\eqar{
\norm{\matU_k\matU_k\transp-\tilde\matU_k\tilde\matU_k\transp}_2
&=&
\norm{\matU_k\matU_k\transp-\matQ\matQ\transp}_2\nonumber
= \norm{ \matQ\transp \matU_k^{\perp}}_2
\\
&=&
\norm{(\matU_k^{\perp})\transp \matQ}_2
=
\norm{\matU_{\rho-k}\transp\matQ}_2.\ \ \ \ \ 
\label{app:GV12}
}
The last equality is because 
\math{\matU_{m-\rho-k}\transp\matQ=\bm0} because \math{\matQ} is in the
range of \math{\matA}.
We now bound \math{\norm{\matU_{\rho-k}\transp\matQ}_2}.
\eqar{
\norm{\matU_{\rho-k}\transp\matQ\matR}_2
&\ge&
\norm{\matU_{\rho-k}\transp\matQ}_2\sigma_k(\matR)\nonumber\\
&\ge&
\norm{\matU_{\rho-k}\transp\matQ}_2\sigma_k^{2p+1}\sigma_k(\matV_k\transp
\matS).
\label{app:chain1}
\\
\norm{\matU_{\rho-k}\transp\matQ\matR}_2
&=&
\norm{\matSig_{\rho-k}^{2p+1}\matV_{\rho-k}\transp\matS}_2\nonumber\\
&\le&
\sigma_{k+1}^{2p+1}\sigma_1(\matV_{\rho-k}\transp\matS).
\label{app:chain2}
}
\r{app:chain1} follows using \r{app:basic4} then
\r{app:basic3} then \r{app:basic2};
\r{app:chain2} uses \r{app:basic1} and submultiplicativity.
Using \r{app:GV12} with \r{app:chain1} and \r{app:chain2}, 
we have:
\begin{lemma}\label{lem:app1}
For any matrix $\matS \in \R^{n \times k}$, 
$$
\sigma_k(\matV_k\transp\matS)
\norm{\matU_k\matU_k\transp-\tilde\matU_k\tilde\matU_k\transp}_2
\le\gamma_k^{2p+1}\sigma_1(\matV_{\rho-k}\transp\matS).
$$
\end{lemma}
Lemma~\ref{lem:app1} holds for general \math{\matS}. We now use the
fact that \math{\matS} is a matrix of i.i.d. standard Gaussians.
Then, for any orthonormal matrix \math{\matV}, 
\math{\matV\transp\matS} is a matrix of i.i.d standard Gaussians.
So, \math{\matV_k\transp\matS} is a \math{k\times k} matrix to which 
Lemma~\ref{lem:gauss2} applies. 
Let \math{\matV \in \R^{n \times n}} be the extension of \math{\matV_{\rho-k} } to a full 
orthonormal
basis.
Then, \math{\matV\transp\matS} is an \math{n\times k} matrix 
to which Lemma~\ref{lem:gauss1} applies (we set \math{t=4}). 
By a union bound, with probability at least 
\math{1-e^{-2n}-2.35\delta}, both inequalities hold:
\eqan{
\sigma_k(\matV_k\transp\matS)&\ge&\delta k^{-\frac12};\\
\sigma_1(\matV_{\rho-k}\transp\matS)&\le&\sigma_1(\matV\transp\matS)
\le4n^{1/2}
}
Using Lemma~\ref{lem:app1} we conclude:
\mand{
\norm{\matU_k\matU_k\transp-\tilde\matU_k\tilde\matU_k\transp}_2
\le
4\gamma_k^{2p+1}\delta^{-1}\sqrt{nk}
\le
4\gamma_k^{2p}\delta^{-1}n.
}
(We used 
\math{\gamma_k\le1} and \math{k\le n}.) Set 
\math{4\gamma_k^{2p}\delta^{-1}n=\varepsilon} and solve for \math{p} to get
\math{p=\ln(\varepsilon\delta/4n)/\ln(\gamma_k^2)}, as desired.
\qed

\end{document}